\documentclass{sig-alternate}

\usepackage[usenames,dvipsnames]{xcolor}
\usepackage{tikz,
			pgflibraryshapes}
\usepackage{latexsym,
			optparams}
\usepackage{array,
			xspace,
			multirow,
			hhline,
			graphicx,
			tabularx,
			booktabs,
			fixltx2e}
\usepackage{ifthen}
\usepackage{enumerate}
\usepackage{comment}
\usepackage{ragged2e}
\usepackage{mathrsfs}
\usepackage{balance}
\usepackage{mathrsfs}
\usepackage{eucal}

\usepackage{lipsum}
\usepackage{multicol}

\newtheorem{theorem}{Theorem}% 
\newtheorem{lemma}{Lemma}% 
\newtheorem{claim}{Claim}% 

\usepackage{ifthen}
\usepackage{enumerate}
\usepackage{comment}

\usepackage[numbers]{natbib}
\usepackage{optparams,ifthen,xspace}

\makeatletter
\let\citeptemp\citep
\long\def\citep@[#1][#2]#3{%
	\ifthenelse{\equal{#1}{}}{%
		\ifthenelse{\equal{#2}{}}{%
			\citeptemp[][]{#3}}{%
			\citeptemp[][#2]{#3}}}{%
		\ifthenelse{\equal{#2}{}}{%
			(#1\citeptemp[][]{#3})}{%
			(#1\citeptemp[][]{#3}, #2)}}}
\renewcommand{\citep}{\optparams{\citep@}{[][]}}
\makeatother

\def\newblock{\hskip .11em plus .33em minus .07em}
\renewcommand{\thanks}[1]{\titlenote{#1}}
\usepackage[textwidth=20mm]{todonotes}

\usepackage{fixltx2e}
\usepackage{stfloats}
\usepackage{wrapfig}
\usepackage{algorithm}
\usepackage{algorithmicx}
\usepackage{algcompatible}
 \algnewcommand\algorithmicreturn{\textbf{return }}
 \algnewcommand\RETURN{\State \algorithmicreturn}%
\usepackage{booktabs}  %% nice tables
\usepackage{verbatim,ifthen}
\usetikzlibrary{arrows}
\usetikzlibrary{decorations.pathreplacing}
\usetikzlibrary{patterns}
\usepackage{pgflibraryshapes}
\usetikzlibrary{positioning,chains,fit,shapes,calc}
\usepackage{tkz-graph}
\usepackage{pgf}
\usepackage{paralist}

\usepackage{mathtools}

\usepackage{booktabs}  %% nice tables
\usepackage{accents}  %% for \mc
\usepackage{enumitem}
\usepackage{mathrsfs}
\usepackage{tikz}
\usetikzlibrary{arrows}
	\newtheorem{remark}[theorem]{Remark}

	\usetikzlibrary{trees}

	\newlength{\wordlength}
	\newcommand{\wordbox}[3][c]{\settowidth{\wordlength}{#3}\makebox[\wordlength][#1]{#2}}

	\newcommand{\nbh}[1][]{
		\ifthenelse{\equal{#1}{}}{\nu}{\nu(#1)}
	}

	\newcommand{\cstr}[1][]{
		\ifthenelse{\equal{#1}{}}{\pi}{\cstr(#1)}
	}

\newcommand\eat[1]{}

\usepackage{lipsum}
\usepackage{multicol}
\usepackage{blindtext}

\begin{document}
\toappear{}

\sloppy

\title{A Discrete and Bounded Envy-Free Cake Cutting Protocol for Four Agents}

\numberofauthors{2}
\author{
%\alignauthor
Haris Aziz \qquad Simon Mackenzie \\
       \affaddr{Data61 and UNSW}\\
       \affaddr{Sydney, Australia}\\
       \email{\{haris.aziz, simon.mackenzie\}@data61.csiro.au}
}

\maketitle

\begin{abstract}
We consider the well-studied cake cutting problem in which the goal is to identify an envy-free allocation based on a minimal number of queries from the agents.
The problem has attracted considerable attention within various branches of computer science, mathematics, and economics. Although, the elegant Selfridge-Conway envy-free protocol for three agents has been known since 1960, it has been a major open problem to obtain a bounded envy-free protocol for more than three agents. 
The problem has been termed the central open problem in cake cutting. We solve this problem by proposing a discrete and bounded envy-free protocol for four agents.   
\end{abstract}




	\section{Introduction}
	\label{intro}

	%\iffalse

	Cake cutting is a metaphor for the allocation of a heterogeneous divisible good among multiple agents with possibly different preferences over different parts of the cake.
	Its main application is fair scheduling, resource allocation, and conflict resolution~\citep{DQS12a} and hence it has been extensively studied within computer science~\citep{Proc15a} and the social sciences~\citep{Thom07a}. Since various important divisible resources such as time and land can be captured by cake cutting, the problem of fairly dividing the cake is a fundamental one within the area of fair division and multiagent resource allocation~\citep{BrTa96a,Guo15a,Proc12a,RoWe98a,Stew99a,Su99a,Thom07a}. 
%	It applies to many settings including allocating rent among housemates, dividing disputed land between land-owners, splitting territory, scheduling police patrol operations, and resolving sea water disputes~\citep[see \eg][]{BrTa96a,Guo15a,Proc12a,RoWe98a,Stew99a,Su99a,Thom07a}.
	%The problem has received attention from diverse communities and a number of textbooks have been written on the topic~\citep[see \eg][]{BrTa96a,RoWe98a}. 

	Formally speaking,  a cake is represented by an interval $[0,1]$ and each of the $n$ agents has a value function over pieces of the cake that specifies how much that agent values a particular subinterval. The main aim is to divide the cake fairly. In particular, an allocation should be envy-free so that no agent prefers to take another agent's allocation instead of his own allocation. 
	Although an envy-free allocation is guaranteed to exist even with $n-1$ cuts~\citep{Su99a}\footnote{The existence of an envy-free cake allocation can be shown via an interesting connection with Sperner's Lemma~\citep{Su99a}.}, \emph{finding} an envy-free allocation is a challenging problem which has been termed ``one of the most important open problems in 20th century mathematics'' by Garfunkel~\citep{Garf88a}.

	\paragraph{Motivation and Contribution}

Unlike allocation of indivisible items~\cite{CoGz15a}, the number of possible allocations in cake cutting is infinite. Since the valuations of agents over the subsets of the cake can be complex, it is not practiceable to elicit each agent's complete valuations function over the cake. A natural approach in cake cutting protocols to query agents about their valuations of different portions of the cake and based on these queries, propose an allocation. A cake cutting protocol is \emph{envy-free} if each agent is guaranteed an envy-free piece if he reports his real valuations.

	For the case of two agents, the problem has a well-known solution in the form of the \emph{Divide and Choose} protocol: one agent is asked to cut the cake into equally preferred pieces and the other agent is asked to choose the preferred piece. The protocol even features in the Book of Genesis (Chapter 13) where Abraham divides the land of Canaan and Lot chooses first. 
	In modern times, the protocol has been enshrined in the Convention of the Law of the Sea~\citep[Page 10, ][]{BrTa96a}. 
	For the case of three agents, an elegant and bounded protocol was independently discovered by John L. Selfridge and John H. Conway around 1960~\citep[Page 116, ][]{BrTa96a}. Since then, an efficient envy-free protocol for four or more agents has eluded mathematicians, economists, and computer scientists.
	%%%%%

	%\footnote{When Conway found his solution, he thought this newsworthy enough that he wrote the solution in a letter to Martin Gardner~\citep[Page 71, ][]{Robe15a}.}

	 In 1995, Brams and Taylor~\citep{BrTa95a} made a breakthrough by presenting an envy-free protocol for \emph{any} number of agents~\citep{Hive95a}. Although the protocol is guaranteed to terminate in finite time, there is one critical drawback of the protocol: the running time or number of queries and even the number of cuts required is unbounded even for four agents. In other words, the number of queries required
	to identify an envy-free allocation can be arbitrarily large for certain valuations functions. If a protocol is not bounded, then its practicality is compromised~\citep{LiRo09a}.   
    %stress the importance of developing finite bounded protocols especially when the aim is to apply them to real world problems.
    Procaccia~\citep{Proc15a} terms unboundedness as a ``serious flaw''.
	Brams and Taylor were cognizant of their protocol's drawback and explicitly mentioned the problem of proposing a bounded envy-free protocol even for $n=4$. 
	Lindner and Rothe~\citep{LiRo09a} write that ``even for $n = 4$, the development of finite bounded envy-free cake-cutting protocols still appears to be out of reach, and a big challenge for future research.'' 
	The problem has remained open and has been highlighted in several works~\citep{BaTa95a,BrTa95a,BrTa96a,BKM05a,EdPr06a,HHA15a,KLP13a,Proc12a,Proc15a,LiRo15a,RoWe98a,SaWa09a}. %MIT09a
	Saberi and Wang~\citep{SaWa09a} term the problem as ``one of the most important open problems in the field'' and Lindner and Rothe~\citep{LiRo15a} mention the case for $n=4$ as ``the central open problem in the field of cake-cutting''. %Procaccia~\citep{Proc12a} goes one step ahead and calls it an interesting challenge within theoretical computer science: ``Since the 1940s, the computation of envy-free cake divisions has baffled many great minds across multiple disciplines. Settling this problem once and for all is an important challenge for theoretical computer science.'' 
	In this paper, \emph{we present a discrete envy-free protocol for four agents that requires a bounded number of queries as well as cuts of the cake}. The maximum number of cuts required is 203.\footnote{Most of the 203 cuts are technically trims but in the Robertson and Webb model, any marking/trim on the cake is also treated as a proper cut.} Some of the techniques we use may be useful for cake cutting protocols with other properties or for more agents.	
In particular, we propose a new technique (called \emph{permutation}) in which by suitably reallocating portions of a partial allocation that is envy-free, we ensure that some agent will not be envious of another agent \emph{even} if the unallocated cake is given to the latter agent. 

	% \citet{SaWa09a}``Finding a bounded algorithm for general n is known to be one of the most important open problems
	% in the field. Even for n = 4 no such algorithm is known.''

	%WoSg07a

	\paragraph{Related Work}

	Cake cutting problems originated in the 1940's when famous mathematicians such as Banach, Knaster, and Steinhaus initiated serious mathematical work on the topic of fair division.\footnote{Hugo Steinhaus presented the cake cutting problems to the mathematical and social science communities on Sep. 17, 1947, at a meeting of the Econometric Society in Washington, D.C.~\citep{RoWe98a,Stei48a}.}
	Since then, the theory of cake cutting algorithms has become a full-fledged field with at least three books written on the topic~\citep{Barb05a,BrTa96a,RoWe98a}. 
	The central problem within cake cutting is finding an envy-free allocation~\citep{GaSt58a,Stew99a}.

	Since the earliest works, mathematicians have been interested in the complexity of cake cutting. Steinhaus~\citep{Stei48a} wrote that ``Interesting mathematical problems arise if we are to determine the minimal number of cuts necessary for fair division.''
	When formulating efficient cake cutting protocols, a typical goal is to minimize the number of cuts while ignoring the number of valuations queried from the agents. In principle, the actual complexity of a problem or a protocol depends on the number of queries. 
	When considering how efficient a protocol is, it is useful to have a formal query model for cake-cutting protocols. Robertson and Webb~\citep{RoWe98a} formalized a simple query model in which there are two kinds of queries: {\sc Evaluate} and {\sc Cut}. In an {\sc Evaluate} query, an agent is asked how much he values a subinterval. In a {\sc Cut} query, an agent is asked to identify an interval, with a fixed left endpoint, of a particular value. 
	Although, the query model of Robertson and Webb is very simple, it is general enough to capture all known protocols in the literature. Note that if the number of queries is bounded, it implies that the number of cuts is bounded in the Robertson and Webb model. The protocol that we present in this paper uses a bounded number of queries in the  Robertson and Webb model. Cake cutting protocols also provide an interesting connection between the literature on fair division and the field of communication complexity~\citep{KuNi06a}.

	There is not too much known about the existence of a bounded envy-free protocol for  $n\geq 4$ except that 
	any envy-free cake-cutting algorithm requires $\Omega(n^2)$ queries in the Robertson-Webb model~\citep{Proc09a,Proc15a}.
	Also,  for $n\geq 3$, there exists no finite envy-free cake-cutting algorithm that outputs \emph{contiguous} allocations~\citep{Stro08a}. 
	Brams et al.~\citep{BTZ97a} and Barbanel and Brams~\citep{BaBr04a} presented envy-free protocols for four agents that require 13 and 5 cuts respectively. However, the protocols are not only unbounded but not even finite since they are \emph{continuous} protocols that require the notion of a \emph{moving knife}.
	 An alternative approach is to consider known bounded protocols and see how well they perform in terms of envy-freeness~\citep{LiRo09a}. Apart from the unbounded Brams and Taylor envy-free protocol for $n$ agents, there are other general envy-free protocols by Robertson and Webb~\citep{RoWe97a} and Pikhurko~\citep{Pikh00a} that are also unbounded. 
	% \footnote{\citet[Page 68, ][]{RoWe98a} write ``Can we produce a finite bounded algorithm for four or more persons such as Selfridge had provided for three? We are waiting for someone to have a good insight.''} 
 
	There are positive algorithmic results concerning envy-free cake cutting when agents have restricted valuations functions~\citep{DQS12a,Bran15b} or when some part of the cake is left unallocated~\citep{SaWa09a}.
	There has also been work on \emph{strategyproof} cake cutting protocols for restricted valuation functions~\citep{AzYe14a,CLPP12a,MaNi12a} as well as strategic aspects of protocols~\citep{BrMi15a}.

	%\fi

	\paragraph{Structure of the Paper}
	In Section~\ref{prel}, the formal model is presented. In Section~\ref{section:warmup}, we present an envy-free protocol for three agents that serves as a warm-up for the case of four agents.
	In Section~\ref{section:4agents}, we presents the main protocol. The section is divided into subsections in which three different protocols (\emph{Post Double Domination Protocol}, \emph{Core Protocol}, and \emph{Permutation Protocol}) are described. These three protocols are used as building blocks to formulate the overall protocol. % (see Figure~\ref{fig:subprotocols}). 
	%Finally, we conclude in Section~\ref{section:disc}.

	% \begin{figure}
	% 	\centering
	% 			\scalebox{0.8}{
	%  \begin{tikzpicture}[%
	%     auto,
	%     block/.style={
	%       rectangle,
	%       draw=blue,
	%       thick,
	%       %fill=blue!20,
	%       text width=6em,
	%       align=center,
	%       rounded corners,
	%       minimum height=2em
	%     },
	%     block1/.style={
	%       rectangle,
	%       draw=blue,
	%       thick,
	%       %fill=blue!20,
	%       text width=9em,
	%       align=center,
	%       rounded corners,
	%       minimum height=2em
	%     },
	%     line/.style={
	%       draw,thick,
	%       -latex',
	%       shorten >=2pt
	%     },
	%     cloud/.style={
	%       draw=red,
	%       thick,
	%       ellipse,
	%       fill=red!20,
	%       minimum height=1em
	%     }
	%   ]
	%  \centering
	%     \draw (3,-1.5) node[block] (M) {Main Protocol};
	%     \path (0,-3) node[block] (C) {Core Protocol}
	%           (3,-3) node[block] (P) {Permutation Protocol}
	% 		    (6.5,-3) node[block1] (PDD) {Post Double~Domination Protocol};
	%
	% 			 \path[draw,thick,->] (M) -- (C);
	% 			  \path[draw,thick,->] (M) -- (P);
	% 			   \path[draw,thick,->] (M) -- (PDD);
	%
	%   \end{tikzpicture}
	%   }
	%   \label{fig:subprotocols}
	%   \caption{The envy-free protocol for 4 agents relies on 3 protocols.}
	% \end{figure}

	\section{Preliminaries}
	\label{prel}

	\paragraph{Model}
			We consider a cake which is represented by the interval $[0,1]$.
			A \emph{piece of cake} is a finite union of disjoint subsets  of $[0,1]$. %The basic cake cutting setting can be represented by the set of agents $N=\{1,\ldots, n\}$ and their valuations functions, which we will denote as \emph{a profile of valuations}.
	We will make the standard assumptions in cake cutting.
			Each agent in the set of agents $N=\{1,\ldots, n\}$ has his own valuation function over subsets of interval $[0,1]$. The valuations are 
		\begin{inparaenum}[(i)]
			\item \emph{defined on all finite unions of the intervals}; \item \emph{non-negative}: $V_i(X)\geq 0$ for all $X\subseteq [0,1]$; 
			\item \emph{additive}: for all disjoint $X,X'\subseteq [0,1]$, $V_i(X\cup X')=V_i(X)+V_i(X')$; 
			\item \emph{divisible} i.e., for every $X\subseteq [0,1]$ and $0\leq \lambda \leq 1$, there exists $X'\subseteq X$ with $V_i(X')=\lambda{V_i(X)}$.
			\end{inparaenum}
	
		We will call an allocation \emph{partial} if there is some cake that is unallocated. A partial envy-free allocation is a partial allocation that is envy-free.  
			 In order to ascertain the complexity of a protocol, Robertson and Webb presented a computational framework in which agents are allowed to make two kinds of queries: (1) for given $x\in [0,1]$ and $r\in \mathbb{R}^+$, {\sc Cut} query asks an agent to return a point $y\in [0,1]$ such that $V_i([x,y])=r$ (2) for given $x,y \in [0,1]$, {\sc Evaluate} query ask agent to return a value $r\in \mathbb{R}^+$ such that $V_i([x,y])=r$. 
A cake-cutting protocol specifies how
agents interact with queries and cuts.
A protocol is \emph{envy-free} if no agent is envious if he follows the protocol truthfully. All well-known cake cutting protocols can be analyzed in terms of number of queries required to return a fair allocation. A cake cutting protocol is \emph{bounded}  if the number of queries required to return a solution is bounded by a function of $n$ \emph{irrespective} of the valuations of the agents.

	% \citep{CLPP12a}

		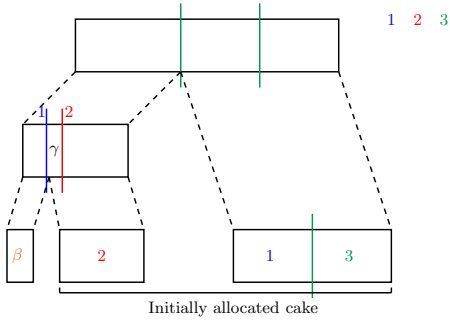
\begin{figure}
			\centering
			\scalebox{0.7}{
		\begin{tikzpicture}[thick]
		\draw (0mm, 0mm) rectangle (50mm, 10mm);
		\node[blue] at (60mm, 10mm) {1};
		\node[red] at (65mm, 10mm) {2};
		\node[ForestGreen] at (70mm, 10mm) {3};

		\draw[ForestGreen] (20mm, 13mm) -- (20mm, -3mm);
		\draw[ForestGreen] (35mm, 13mm) -- (35mm, -3mm);
	
		[fill,red] (0mm, 5mm) rectangle (19.8mm, 10mm);
		[fill,blue] (0mm, 0mm) rectangle (19.8mm, 5mm);
	\draw[dashed](0mm,0mm)--(-10mm,-10mm);
		\draw[dashed](20mm,0mm)--(10mm,-10mm);
		\draw (-10mm,-20mm) rectangle (10mm,-10mm);
	
	\draw[blue] (-5.5mm, -7mm) -- (-5.5mm, -23mm);
		\draw[red] (-2.5mm, -7mm) -- (-2.5mm, -23mm);
		\node[blue] at (-6.4mm, -7.5mm) {1};
		\node[red] at (-1.2mm, -7.5mm) {2};
			\node[] at (-4mm, -15mm) {$\gamma$};
	
	\draw[dashed](-5mm,-20mm)--(-3mm,-30mm);
		\draw[dashed](10mm,-20mm)--(13mm,-30mm);
		\draw (-3mm,-40mm) rectangle (13mm,-30mm);

	\draw[dashed](-5mm,-20mm)--(-8mm,-30mm);
		\draw[dashed](-10mm,-20mm)--(-13mm,-30mm);
		\draw (-13mm,-40mm) rectangle (-8mm,-30mm);

		\node[brown] at (-11mm, -35mm) {$\beta$};
		\node[red] at (5mm, -35mm) {2};
	
	\draw[dashed](20mm,0mm)--(30mm,-30mm);
		\draw[dashed](50mm,0mm)--(60mm,-30mm);
		\draw (30mm,-40mm) rectangle (60mm,-30mm);
		\draw[ForestGreen](45mm,-27mm)--(45mm,-43mm);

		\node[blue] at (37mm, -35mm) {1};
		\node[ForestGreen] at (52mm, -35mm) {3};

		\draw(-3mm,-42mm)--(60mm,-42mm);
		\draw((-3mm,-42mm)--((-3mm,-41mm);
		\draw((60mm,-42mm)--(60mm,-41mm);
		\node[] at(30mm,-45mm){Initially allocated cake};
	
		\end{tikzpicture}
		}
		\caption{Example of a trim. Agents $1$ and $2$ trim their most preferred piece (the left most piece) to the value equal to that of their second most preferred piece. %By trimming, they indicate that the part of the piece on the right hand side of the trim is as preferred as the second most preferred piece.
	 In this instance, agent $2$ trims more than agent $1$, hence his trim is to the right of agent $1$'s trim. 
     Let us assume agent $1$ and $3$ each get a complete piece with $1$ getting his second most preferred piece.   
		%Since agent $1$ trims more, we say that agent $1$ `sacrifices' more.
		Agent $2$ is not envious of other agents if he gets the right side of the trimmed piece up till his trim. 
	If agent $2$ gets the part to the right of agent $1$'s trim instead of his own trim, then he is even happier. We will refer to this extra bit $\gamma$ as the `bonus' for agent $2$.  Agent $1$ is still not envious of agent $2$ if $2$ gets the bonus.}
		\label{fig:trim}
		\end{figure}

	\paragraph{Terms and Conventions}
	We now define some terms and conventions that we will use in the paper. 
	Given a partial envy-free allocation of the cake and an unallocated residue $\beta$, we say that agent $j$ \emph{dominates} agent $i$ if $j$ does not become envious of $i$ even if all of $\beta$ were to be allocated to $i$.  This concept has been referred to as $i$'s \emph{irrevocable advantage} in the cake cutting literature~\citep{BrTa96a}. 

	In the cake cutting protocols, we will describe, an agent may be asked to trim a piece of cake so that its value equals the value of a less valuable piece.  
	 Agents will be asked to trim various pieces of the cake so their remaining value is equal to the value of the third 
(or in some cases their second) most preferred complete piece. In Figure~\ref{fig:trim}, we outline the idea of trimming a piece to equal the value of some other piece.
	 When an agent trims a piece of cake, he will trim it from the left side: the main piece (albeit trimmed) will be on the the right side. The piece minus the  trim will be called the partial main piece. The remainder will be referred to as the \emph{residue}. If an agent trims a piece, we say he is \emph{competing} for the piece. 
	%%%HMM NEED TO RESOLVE    % \haris{Do we need ``properly''}
    %If an agent trims a piece properly, we say he is \emph{competing} for the piece. 
	 When we say an agent is \emph{guaranteed} to get his second/third/etc most favoured piece, this guarantee is based on the \emph{ordinal} preferences of the agents over the pieces. By ordinal, we mean that agents simply give a weak ordering over the pieces but do not tell the exact cardinal utility difference between two pieces. If an agent is indifferent between the top three pieces, then we will still say that the agent is guaranteed to get his third most valued piece.

	%\begin{example}[Trim]
	
		%\end{example}

	%%%%

		We introduce notation to represent which agents have trimmed which pieces. So for example $123|1|2|3$ represents the scenario where
	one piece has three trims (by agents 1,2,3) in any possible order and the other three pieces have one trim each by one of the agents 1, 2, 3. We may enrich this notation further as follows: $1_12_13_1|1|2|3$ which represents that all agents think that the piece with the three trim marks is their most preferred or equivalently highest valued piece.

			\begin{algorithm}[b!]
				% \SetAlgorithmName{MegaAlgorithm}{}
			  \caption{Envy-free Protocol for 3 Agents.}
			  \label{algo:3agents}
			\renewcommand{\algorithmicrequire}{\wordbox[l]{\textbf{Input}:}{\textbf{Output}:}} 
		
			 \renewcommand{\algorithmicensure}{\wordbox[l]{\textbf{Output}:}{\textbf{Output}:}}
	%		\algsetup{linenodelimiter=\,}
	%		\begin{algorithmic}
			%	\REQUIRE 
	%			\ENSURE EF complete allocation.
	%		\end{algorithmic}
			  \begin{algorithmic}[1] 
				  %\normalfont
			%	   \small	
			  	\STATE Agent 3 divides the cake into 3 equally preferred pieces.
			  	\IF{1 and 2 can each be given a different complete most preferred piece}
				\STATE give that complete piece to the agent who prefers it the most and give remaining piece to $3$ and return.
                \ELSE 
			  	\STATE 1 and 2 trim their highest valued piece from the left side to make the right side of the trim equal to the value of second most preferred piece (they simultaneously put trim marks).
                \ENDIF
			\IF{1 and 2 trim the same piece}
	\STATE consider $\beta_1$, the remainder from the  left extreme of the piece to the leftmost trim. The partial piece ${P_1^1}'$  which is the most preferred piece except $\beta_1$ is given to the agent $i\in \{1,2\}$ who trimmed the piece more (let the other agent be $-i$). Let $\gamma_1$ be the part between the two trims of agent 1 and 2. Agent $-i$ gets his second highest valued complete piece ${P_2^1}$. Agent $3$ gets the remaining complete piece. The  unallocated cake is $\beta_1$. %If $\beta_1$ is empty, then return.
	\ENDIF
			  	\STATE 3 cuts $\beta_1$ into 3 equally preferred pieces.
				\STATE 1 and 2 trim their highest valued piece from the left side to make it equal to the value of second best.
			  	\IF{1 and 2 can be each be given a different complete most preferred piece}
				\STATE give that complete piece to the agent who prefers it the most and give remaining piece to $3$ and return.
			  	% \ELSIF{1 and 2 trim the same piece but $-i$ trims more}, \STATE give that piece to $-i$, then give $i$ his second highest valued piece, and 3 gets the remaining complete piece.
					\ELSIF{1 and 2 trim the same piece but $-i$ trims at least as much as $i$}, 
	\STATE give $-i$ the most preferred piece up till the leftmost trim, give $i$ a complete second most preferred piece, and give $3$ the remaining complete piece.
			  	\ELSIF{1 and 2 trim the same piece but $i$ trims more again},
		\STATE let $\beta_2$ be the remainder from the left hand side to the first trim in $\beta_1$. The partial piece ${P_1^2}'$  which is the most preferred piece except $\beta_2$ is given to the agent $i$. Let $\gamma_2$ be the part between the two trims. Agent $-i$ gets his second highest valued complete piece ${P_2^2}$. Agent $3$ gets the remaining complete piece. The only unallocated cake left if any is $\beta_2$.
%         \STATE Since $i$ again got a partial piece, $i$ is asked to identify the lesser preferred $\gamma_j\in \{ \gamma_1, \gamma_2\}$. Then $i$ forgoes ${P_1^j}'$ and in return gets ${P_2^j}$. On the other hand, $-i$ forgoes ${P_2^j}$ and gets in return  ${P_1^j}'$.
          \STATE Since $i$ again got a partial piece, $i$ is asked to identify the lesser preferred $\gamma_j\in \{ \gamma_1, \gamma_2\}$. Then $i$ gives ${P_1^j}'$ to agent $-i$ and gets ${P_2^j}$ in return.
        			\ENDIF
				\STATE If some cake is still unallocated, Agent $1$ and $2$ perform Divide and Choose to allocate it.
			 \end{algorithmic}
			\end{algorithm}
    
    % \section{Protocol for three agents: a warmup}
    	\section{Protocol for three agents}
	\label{section:warmup}
    
    	We warm-up by presenting a protocol (Algorithm~\ref{algo:3agents}) for $n=3$ which we will extend to $n=4$.

Although the protocol requires more cuts than the Selfridge-Conway protocol, it bears similarities with it. It also depends on some ideas that we will exploit for our main protocol for four agents. The main idea of the protocol is that in each step, the cutter (agent 3) cuts the unallocated cake into 3 equally preferred pieces and gets one of the complete pieces. In each step, a partial envy-free allocation is maintained and then the remainder is again allocated which results in the remainder being fully allocated or a smaller remainder left. 
	Note that when $i\neq 3$ is given the partial cake piece, agent $3$ dominates $i$. When the remainder $\beta_1$ in the first step is divided among the agents, if now $3$ dominates $-i$, then $3$ does not care how $\beta_1$ is divided among $i$ and $-i$ since he dominates both. So $3$ is in this sense `eliminated' from the protocol and we can perform Divide and Choose for 1 and 2 on the unallocated cake. If $3$ again dominates $i$ based on how $\beta_1$ is allocated, then we enforce a \emph{permutation} or \emph{reallocation} of some pieces of $i$ and $-i$, so that $3$ dominates $-i$ (see Figure~\ref{fig3agents-perm}). Both the ideas of domination and permutation will feature prominently in our protocol for four agents.

	\begin{figure}
		\begin{center}
			\scalebox{0.7}{
	\begin{tikzpicture}[thick]
	\draw (0mm, 0mm) rectangle (70mm, 10mm);

	\draw[brown] (30mm, 13mm) -- (30mm, -3mm);
	\node[brown] at (15mm, 5mm) {$\beta_1$};
	%\node[orange] at (50mm, 5mm) {$\alpha$};

	\draw[dashed](0mm,0mm)--(-10mm,-10mm);
	\draw[dashed](30mm,0mm)--(20mm,-10mm);
	\draw (-10mm,-20mm) rectangle (20mm,-10mm);

	\draw[dashed](30mm,0mm)--(40mm,-10mm);
	\draw[dashed](70mm,0mm)--(80mm,-10mm);
	\draw (40mm,-20mm) rectangle (80mm,-10mm);

	\draw (50mm, -10mm) -- (50mm, -20mm);
	\draw (65mm, -10mm) -- (65mm, -20mm);

	\node[ForestGreen] at (72mm, -15mm) {3};

	% \node[blue] at (45mm, -15mm) {1};
	% \node[red] at (57mm, -15mm) {2};

	\draw[dotted](40mm,-15mm)--(20mm,-15mm);

	\draw[brown] (0mm, -7mm) -- (0mm, -23mm);
	\node[brown] at (-5mm, -15mm) {$\beta_2$};
	%\node[orange] at (10mm, -15mm) {$\alpha'$};

	\draw[dashed](-10mm,-20mm)--(-20mm,-30mm);
	\draw[dashed](0mm,-20mm)--(-10mm,-30mm);
	\draw (-20mm,-40mm) rectangle (-10mm,-30mm);
	\draw[dashed](0mm,-20mm)--(10mm,-30mm);
	\draw[dashed](20mm,-20mm)--(30mm,-30mm);
	\draw (10mm,-40mm) rectangle (30mm,-30mm);

	\draw (16mm, -30mm) -- (16mm, -40mm);
	\draw (23mm, -30mm) -- (23mm, -40mm);

	\node[ForestGreen] at (27mm, -35mm) {3};

	% \node[blue] at (13mm, -35mm) {1};
	% \node[red] at (20mm, -35mm) {2};

	\node[ForestGreen] at (72mm, -15mm) {3};
	\draw[dotted](-10mm,-35mm)--(10mm,-35mm);

	\node[red] at (45mm, -15mm) {2};
	\node[blue] at (57mm, -15mm) {1};

	\node[red] at (13mm, -35mm) {2};
	\node[blue] at (20mm, -35mm) {1};

	\end{tikzpicture}
	}
	\end{center}
	
	\caption{Protocol for 3 agent: In case agent $2$ gets a trimmed piece two times, then we need to perform a permutation so that $1$ gets a trimmed piece. }
	\label{fig3agents-perm}
	\end{figure}
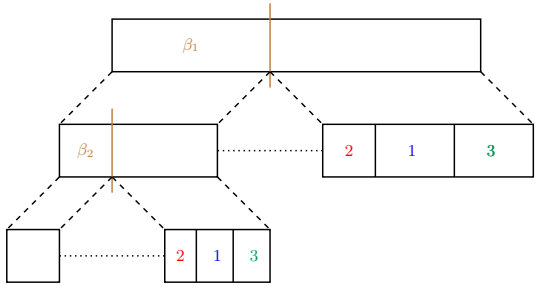

	\section{Protocol for Four Agents}
	\label{section:4agents}

	We now extend the ideas for the case of $n=3$ to $n=4$. 
	The general flavour of the protocol is similar to that of the protocol for 3 agents. A designated cutter is asked to cut into equally preferred pieces. Based on some finer steps, we are able to achieve an envy-free allocation with possibly some cake still unallocated. We repeat the process on the remaining unallocated piece with the goal that the cutter dominates other agents just as we managed in the protocol for 3 agents.
	A few additional complications are introduced when dealing with 4 agents.
	Eliminating an agent would require being able to ensure that we can make an agent dominate all 3 others.
	Thankfully this is not necessary: we can show that we only need a protocol that ensures a given agent dominates 2 others. This is proved in the Double Domination Lemma.

	% \begin{figure}[t]
	% 	\centering
	% 	\begin{tikzpicture}[scale=0.3]
	% 		%\scriptsize
	% 		\centering
	% 		\tikzstyle{pfeil}=[--,draw]
	% 		\tikzstyle{pfeil}=[->,>=angle 60, shorten >=1pt,draw]
	% 		\tikzstyle{onlytext}=[]
	%
	% 		\draw
	% 			node[circle,fill=white,draw](3){$\textcolor{ForestGreen}{3}$}
	% 			++(-45:4)
	% 			node[circle,fill=white,draw](4){$\textcolor{purple}{4}$}
	% 			++(45:4)
	% 			node[circle,fill=white,draw](2){$\textcolor{red}{2}$}
	% 			++(135:4)
	% 			node[circle,fill=white,draw](1){$\textcolor{blue}{1}$}
	% 		 	++(225:4)
	% 		;
	%
	%
	% 		 \path[draw,thick,->] (1) -- (4);
	% 		  \path[draw,thick,<->] (1) -- (3);
	% 		  \path[draw,thick,->] (2) -- (1);
	% 		    \path[draw,thick,->] (2) -- (4);
	% 			 \path[draw,thick,->] (3) -- (2);
	%
	% 	\end{tikzpicture}
	%
	% 	\caption{The domination graph of the final case in the proof of the Double Domination Lemma. }
	% 	\label{fig:dgraph}
	% \end{figure}

	\subsection{Post Double Domination Protocol}

	We now present the Post Double Domination Protocol (Algorithm~\ref{algo:postDD}) that takes as input a partial envy-free allocation in which each agent dominates two agents and it returns a complete envy-free allocation.

			\begin{algorithm}[h!]
				% \SetAlgorithmName{MegaAlgorithm}{}
			  \caption{Post Double Domination Protocol for 4 Agents}
			  \label{algo:postDD}
			\renewcommand{\algorithmicrequire}{\wordbox[l]{\textbf{Input}:}{\textbf{Output}:}} 
		
			 \renewcommand{\algorithmicensure}{\wordbox[l]{\textbf{Output}:}{\textbf{Output}:}}
	%		\algsetup{linenodelimiter=\,}
		\begin{algorithmic}
		%	\small
				\REQUIRE A partial envy-free allocation and unallocated cake such that each agent dominates 2 other agents
				\ENSURE Envy-free complete allocation.
			\end{algorithmic}
			  \begin{algorithmic}[1] 
		%		  \small
			\STATE There exists some agent $1$ who dominates two other agents say $3$ and $4$.% where the 4 agents are $i,j,k,\ell$.  
				  \IF{$2$ also dominates $3$ and $4$}
				  \STATE  $3$ and $4$ can divide the remainder by Divide and Choose.
				  % If $j$ also dominates $k$ and $\ell$, $k$ and $\ell$ can divide the remain by Divide and Choose.
				  \ELSIF{$2$ does not dominate $3$ and $4$}
				  \STATE it dominates $1$ and one of $3$ and $4$ say $4$.  
				  \IF{$3$ dominates $4$}
				  \STATE give all the remainder to $4$ (since everyone dominates $4$).
				  \ELSIF{$3$ does not dominate $4$ and hence dominates $1$ and $2$}
				  \STATE then 
				  let  $4$ cut the residue into four equally preferred pieces, and agents $1,2,3$ pick their most preferred remaining piece in that order.
				  \ENDIF
				  \ENDIF
				  \RETURN Envy-free complete allocation.
			 \end{algorithmic}
			\end{algorithm}

	\begin{figure}[h!]
		\centering
		\begin{tikzpicture}[scale=0.3]
			%\scriptsize
			\centering
			\tikzstyle{pfeil}=[--,draw]
			\tikzstyle{pfeil}=[->,>=angle 60, shorten >=1pt,draw]
			\tikzstyle{onlytext}=[]

			\draw
				node[circle,fill=white,draw](3){$\textcolor{ForestGreen}{3}$}
				++(-45:4)
				node[circle,fill=white,draw](4){$\textcolor{purple}{4}$}
				++(45:4)
				node[circle,fill=white,draw](2){$\textcolor{red}{2}$}
				++(135:4)
				node[circle,fill=white,draw](1){$\textcolor{blue}{1}$}
			 	++(225:4)
			;

			 \path[draw,thick,->] (1) -- (4);
			  \path[draw,thick,<->] (1) -- (3);
			  \path[draw,thick,->] (2) -- (1);
			    \path[draw,thick,->] (2) -- (4);
				 \path[draw,thick,->] (3) -- (2);

		\end{tikzpicture}

		\caption{The domination graph of the final case in the proof of the Double Domination Lemma. }
		\label{fig:dgraph}
	\end{figure}
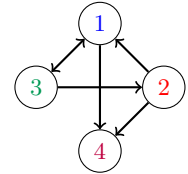
    
		\begin{lemma}[Double Domination Lemma]\label{lemma:dd}
	Suppose we have a bounded protocol which given a specified agent $i$ and an unallocated piece of cake returns a partial envy-free allocation such that $i$ dominates 2 other agents, then we can extend this into a 4 agent envy-free bounded protocol.
			\end{lemma}
			\begin{proof}
				If we have a bounded protocol in which one agent can be made to dominate two other agents, then we simply run it at most 4 times on any unallocated cake to ensure that each agent dominates two other agents. If while doing this, the cake is completely allocated, we are already done. Otherwise, we can run the Post Double Domination Protocol (Algorithm~\ref{algo:postDD}). We now argue for the correctness of the protocol.
			
				Assume $1$ dominates $3$ and $4$. Now if $2$ also dominates $3$ and $4$, then there exists a partial envy-free allocation in which even if all the residue is given to $3$ or $4$, then agent $1$ and $2$ will not be envious. 
	All the residue can be divided among $3$ and $4$ using divide and choose. Agents $1$ and $2$ don't care because they dominate $3$ and $4$.

	The other case is when $2$ does not dominate $3$ and $4$. Without loss of generality assume that $2$ dominates $1$ and $4$. Now if $3$ dominates $4$, we are already done 
	because the whole residue can be given to $4$ since everyone dominates $4$. If $3$ does not dominate $4$ but dominates $1$ and $2$, then the domination graph looks like in Figure~\ref{fig:dgraph}, an envy-free allocation can be found
	via the following method: agent $4$ cuts the residue into equally preferred four pieces, and agents $1,2,3$ pick their most preferred remaining piece in that order. Agent $2$ dominates $1$ so does not care if $1$ chooses first; agent $3$ dominates $2$ and $1$ so that he does not care if $1$ and $2$ choose before him. 		
	\end{proof}

	Since we have shown that making an agent dominate two other agents is helpful, we will now explain how to achieve it. The overall protocol will first achieve double domination for each agent and if some cake is still unallocated, it will use the Post Double Domination Protocol. In order to get an agent to dominate other agents, we will repeat a \emph{core protocol} multiple times which gives a partial envy-free allocation. The core protocol is explained in the next section.

	\subsection{Core Protocol}

	%\begin{minipage}%{\dimexpr.5\textwidth-.5\columnsep}
	%\begin{minipage}{1\columnwidth}

	 In the core protocol, a specified agent is asked to cut the cake into equally preferred pieces. The cutter gets a complete piece whereas the other agents may get partial pieces. 
	The core protocol is recursively applied to the unallocated cake. After a bounded number of calls of the core protocol, we are in a position to do some reallocation so as to ensure that the specified cutter dominates two agents. We can then repeat this for another specified agent until each agent dominates two other agents.

	We now give a high level description of the \emph{core protocol}. Let us say that agent $4$ divides the unallocated cake into $4$ pieces. %$S_1,S_2,S_3,S_4$. 
	Agents $1$, $2$, and $3$ are asked to trim  the  left hand side of their most preferred two pieces to make the right side of their trim equally valuable as their third most preferred cake piece. Each agent in set $\{1,2,3\}$ trims at most two pieces. In case an agent is indifferent between two or more pieces, we will still assume that the agent trims one piece to make it equal to the other piece. The trim in this extreme case is a trivial trim that coincides with the left edge of the cake. Hence the four pieces have a total of six trims.
	If an agent who trims a piece most is given that piece (an agent who trims most two pieces can choose which piece to get), up to his trim point, then the agent is envy-free. In fact the agent is not envious even if each other agent who gets a piece is given the piece up till the second rightmost trim. % \footnote{This idea of the piece being given to the agent who trimmed the most but the piece being given up till the second most trim is reminiscent of the second price auction in which the resource is allocated to the agent who bid the most but he pays the price of the second highest bid.}
     This approach is useful to get a partial envy-free allocation with some cake unallocated. In the core protocol, we do some extra work so that apart from the cutter, at least one more agent is given a complete piece. In order to do this, in a couple of cases, we may ask one or two carefully identified agents to additionally trim their most preferred piece to the value of their second most preferred piece in which case the previous trims of these agents are ignored. This is helpful is ensuring that at least two agents get complete pieces.  
	It may be the case that some cake is left unallocated in which case the core protocol may be implemented on the unallocated cake again. The main thing we will prove is that we do not need to implement the core protocol on the unallocated cake unbounded number of times. For this, we first show in the Core Protocol Lemma that the core protocol returns a partial envy-free allocation with the additional useful property that the cutter and one other agent get complete pieces.

			\begin{algorithm}[h!]
				% \SetAlgorithmName{MegaAlgorithm}{}
			  \caption{Core Protocol for 4 agents that returns a partial envy-free allocation}
			  \label{algo:core}
			\renewcommand{\algorithmicrequire}{\wordbox[l]{\textbf{Input}:}{\textbf{Output}:}} 
		
			 \renewcommand{\algorithmicensure}{\wordbox[l]{\textbf{Output}:}{\textbf{Output}:}}
	%		\algsetup{linenodelimiter=\,}
		\begin{algorithmic}
		%\small
				\REQUIRE Specified cutter agent (say agent 4)
		\ENSURE Partial envy-free allocation.
			\end{algorithmic}
			\begin{minipage}{1\columnwidth}
			  \begin{algorithmic}[1] 
				  %\normalfont
		%	 \small	
				% \footnotesize
				%\small
			  	\STATE  Agent 4 is asked to cut the cake into 4 equal value pieces.
				\STATE Agents $1,2,3$ are asked to give their values for the 4 pieces. 
				\IF{each agent in $\{1,2,3\}$ can be given a most preferred piece}
				\STATE Allocate each agent in $\{1,2,3\}$ a most preferred (complete) piece and the remaining complete piece to the cutter (agent 4). 
				\RETURN the envy-free allocation.
				\ENDIF
				\STATE Agents $1,2,3$ are asked to trim the left hand side of their most and second most preferred pieces to make them (the right side of the trim) equally valuable as their third most preferred cake piece. 
												 \IF{no piece has exactly one trim}

				\IF{we are in a case $ij|jk|ik$ where $\{i,j,k\}=\{1,2,3\}$.}
												 \STATE  %We are in a case $ij|ik|ik$ where $\{i,j,k\}=\{1,2,3\}$. 
												 Since we have already covered the case in which each agent can be given a complete most preferred piece, we end up in situation $i_1j_1|j_2k|i_2k|$. Without loss of generality, the case is $i_1j_1|j_2k_1|i_2k_2|$. In this case $i$ and $k$ are asked to trim their most preferred piece to their second most preferred piece whereas agent $j$ is asked to trim his two most preferred pieces to equal his third most preferred piece. Agent $i$ and $k$'s trims up to their third most preferred piece are ignored. The effective trims look as follows now: $ij|jk|||$. 
								 
									\IF{$j$ does not have the rightmost trim in both pieces}			 
				%If $j$ does not have the rightmost trim in both pieces,  then 
				the right side of each piece with two trims is given to the agent who trimmed it the most. The piece is given up till the second rightmost trim.  The remaining agent picks his most preferred complete unallocated piece and then 4 get the remaining unallocated piece.
				\ELSIF{ $j$ trimmed both the pieces the most} 
				%If $j$ trimmed both the pieces the most, then 
				$j$ can choose which piece (up till the second rightmost trim) to get. The other piece with the trims is given to the agent with the second rightmost trim up till the second rightmost trim. The third non-cutter $i$ or $k$ gets his second most preferred piece completely. The last unallocated complete piece is given to agent 4. 
				\ENDIF
												% \ELSIF{GG}
						% 						\STATE
						 						\ENDIF
						\ELSIF{we are in a case $ijk|ijk|||$ where $\{i,j,k\}=\{1,2,3\}$.}
						\STATE Ask each agent in $\{1,2,3\}$ to trim this first and second most preferred piece from the left side to make the right hand side to the value of his third most preferred piece. Each agent is given the piece he trims the most up till the second rightmost trim. If the same agent has the rightmost trims for both pieces, he chooses which piece to get. The agent with second rightmost trim gets the other piece up till the second rightmost trim. The third non-cutter gets his third most piece completely. The last unallocated complete piece is given to agent 4. 
												\ENDIF
							\COMMENT{Continued on next page...}

							\algstore{myalg}
			 \end{algorithmic}
			 \end{minipage}
	\end{algorithm}

	\begin{lemma}[Core Protocol Lemma]
	For $n=4$, there exists a discrete and bounded protocol which returns partial envy-free allocation in which one agent cuts the cake into four equally preferred pieces and the cutter as well as at least one other agent gets one of these four pieces. 
		\end{lemma}
			\begin{proof}
	We argue that when agent 4 cuts the cake into equally preferred piece and then these pieces are partially allocated to the agents, then (1) the partial allocation is envy-free, and (2) the cutter and one other agent get complete pieces.
	If each agent in $\{1,2,3\}$ can be given a most preferred piece, then both conditions are trivially met. Otherwise, the algorithm distinguishes between the following cases:
	\begin{inparaenum}[(1)]
		\item no piece has exactly one trim;
		\item exactly one piece has exactly one trim;
		\item exactly two pieces have exactly one trim;
		\item exactly three pieces have exactly one trim.
	\end{inparaenum}
				In each of the cases, one non-cutter gets a complete piece and the cutter is also given an unallocated complete piece. 
			
				It remains to be shown that the partial allocation is envy-free. When an agent $i$ gets a (possibly partial) piece $a$, he was the one who trimmed that piece the most. For each other piece $b$ that is allocated, some other agent $j$ trimmed $b$ at least as much as agent $i$, i.e.,  $j$'s trim in $b$ was not left of $i$'s trim. Hence $i$ is not envious of $j$ if $j$ gets $b$ up till $j$'s trim from the right hand side. The reason is that $i$ thinks that $j$'s piece has more value than $i$'s allocation only if $j$ gets the right side of $b$ beyond the trim of agent $i$.
	Thus, if $i$ has the second rightmost trim in $b$, then $i$ is not envious of $j$ even if $j$ gets the right side of $b$ up till $i$'s trim (e.g., see Figure~\ref{fig:trim2}). We note that in each of the four cases, when an agent $i$ get a piece, he is not envious of another agent $j$ because of the reason above. Moreover, if in piece $a$, the second rightmost trim is strictly to the left of $i$'s trim in $a$ or if there is no other trim in $a$, then $i$ not only gets the right hand side of $a$ up to $i$'s trim but an additional bonus up till the second rightmost trim or the edge of the cake (whichever comes first). Hence, no agent $i$ is envious of another agent $j$.
	\end{proof}

	\begin{remark}
		 As soon as the agents' ordinal ranking of the four pieces cut by the cutter are known, it can be ascertained whether in the core protocol, an agent is guaranteed to get a piece of value equal to his third or second most preferred piece. Each agent gets a piece that is of same value as his third most preferred piece. 
		 An agent is guaranteed to get a second most preferred piece during the core protocol, if he is asked in the worst case to trim his most preferred piece to his second most preferred piece.
		 %Note that initially all agents are asked to trim their higher preferred pieces upto the value of their third most preferred piece. In each case of the core protocol, each agent at gets a piece of value that is at least his third most preferred piece. For example, if the trims are $12|13|23|$, then, we know that in the worst case, two other agents may get partial pieces for the agent $1$'s two most preferred pieces. Even then, agent $1$ can get a third most preferred piece. 
	%If $i$ is guaranteed to get at least a second most preferred piece. In this case, his previous trims are ignored and he is asked to trim a most preferred piece to equal the value of a second most preferred piece. For example, assume that the trims to equal the third most preferred piece are $1|123|23|$, then in the worst case 2 or 3 get a partial piece of agent $1$'s most preferred piece. Even then agent $1$ is guaranteed to get a piece equal to his second most preferred piece.
		\end{remark}

				\begin{figure}[h!]
					\centering
					\scalebox{0.75}{
				\begin{tikzpicture}[thick]
				\draw (0mm, 0mm) rectangle (80mm, 10mm);
				% \node[blue] at (60mm, 15mm) {1};
		% 		\node[red] at (65mm, 15mm) {2};
		% 		\node[ForestGreen] at (70mm, 15mm) {3};
		% 		\node[purple] at (75mm, 15mm) {4};

				\draw[brown] (20mm, 15mm) -- (20mm, -5mm);
				\draw[brown] (45mm, 15mm) -- (45mm, -5mm);
				\draw[brown] (65mm, 15mm) -- (65mm, -5mm);
		
				\draw[blue] (5mm, 13mm) -- (5mm, -3mm);
				\node[blue] at (5mm, -5mm) {$i$};	
		
				\draw[red] (27mm, 13mm) -- (27mm, -3mm);
				\node[red] at (27mm, -5mm) {$j$};
		
				\draw[blue] (23mm, 13mm) -- (23mm, -3mm);
				\node[blue] at (23mm, -5mm) {$i$};
		
				\node[] at (10mm, 17mm) {$a$};
					\node[] at (32.5mm, 17mm) {$b$};

				[fill,red] (0mm, 5mm) rectangle (19.8mm, 10mm);
				[fill,blue] (0mm, 0mm) rectangle (19.8mm, 5mm);
	
				\end{tikzpicture}
				}
				\caption{Example of a scenario where $i$ has the rightmost trim for piece $a$ and second rightmost trim of piece $b$ whereas agent $j$ has the rightmost trim for piece $b$. If $i$ gets the right hand side of $a$ up till his trim, and $j$ gets the right hand side of $b$ till $i$'s trim in $b$, then $i$ is not envious of $j$'s allocation. }
				\label{fig:trim2}
				\end{figure}
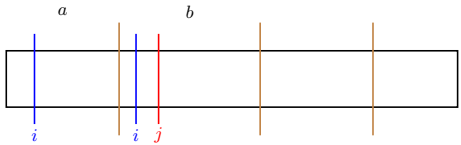

	\begin{algorithm}[h!]
						% \SetAlgorithmName{MegaAlgorithm}{}
					 % \caption{Core Protocol for 4 agents that returns a partial envy-free allocation}
					 % \label{algo:core}
					\renewcommand{\algorithmicrequire}{\wordbox[l]{\textbf{Input}:}{\textbf{Output}:}} 
		
					 \renewcommand{\algorithmicensure}{\wordbox[l]{\textbf{Output}:}{\textbf{Output}:}}
			%		\algsetup{linenodelimiter=\,}
				% \begin{algorithmic}
	% 			%\footnotesize
	% 			\small
	% 					\REQUIRE Specified agent (say agent 4)
	% 			\ENSURE Partial envy-free allocation.
	% 				\end{algorithmic}
	\begin{minipage}{1\columnwidth}
					  \begin{algorithmic}[1]
						  \algrestore{myalg} 
						  %\normalfont
						%   \small
						%\footnotesize
					
												\IF{exactly one piece has exactly one trim}
												\COMMENT{The trims look like $i|ijk|jk|$}
												\IF{the agent who trimmed it views it as his most preferred}
												\STATE Give the complete piece to him.
												\ELSIF{the agent who trimmed it find it his second most preferred}
												\STATE ask him to trim his first most preferred to equal his second most preferred piece.
												\ENDIF
											    % \STATE % (let's assume it was agent 3)
											     %\STATE If it is his second most preferred, ask him to trim his first most preferred to equal his second most preferred piece,
						\STATE	%Trim each piece up to their second rightmost trim.
						The pieces are allocated up to the second rightmost trim to the agents who trimmed them most. If 2 pieces were trimmed most by the same agent, he decides which to get and the other is given to the agent who trimmed that piece second most.
												  % \STATE Give the piece with zero trims completely to the cutter (agent 4).
											     						 \STATE Give the last unallocated piece completely to the cutter (agent 4).
																		 \ENDIF
					
	 					\IF{exactly 2 pieces have exactly one trim }
						\COMMENT{The trims look like $jk|ik|i|j$}
	 					 \STATE Give those pieces with exactly one trim completely to the agents in $\{1,2,3\}$ who trimmed them if it is their most preferred.
	 					 \STATE Agents who trimmed a piece with a single trim but do not value that piece most are asked to re-trim their most preferred piece up to their second most preferred piece. Their trims to make them equal to their third most preferred ignored from now on.
	 					 % \haris{are their trims to make them equal to their most preferred ignored from now on?}
	 					 \STATE The right hand side of the two pieces with the two trims are given to the agents with the rightmost trims. The pieces are allocated up till the second rightmost trim. If the 2 pieces were trimmed most by the same  agent (agent $k$), he choses which to get and the other piece is given to whoever trimmed it second most. 
						 \STATE If one non-cutter has not been allocated a piece, he gets the most preferred piece among the two unallocated complete pieces.
						 %
						 % \textbf{$jk|ik||$
						 % If two different agentd trimmed most, they get partial pieces. The remaining agents get }
						  					 \STATE Give the last unallocated piece completely to the cutter (agent 4)
	 					\ENDIF.
					
	 						\IF{exactly 3 pieces have exactly one trim}
							\STATE The trims look like $1|2|3|123$.
	 						 \STATE If an agent most prefers the piece where he made a single trim, give him that complete piece.
	 						 \STATE The ones who most prefer the piece with the three trims (call it $a$) compete for it by trimming up to their second most preferred piece. 
							 % \haris{clarify the most prefer versus prefer in bold}
							 Their trims to make them equal to their most preferred ignored from now on.

	 						 \STATE Cut $a$ at the second rightmost trim, and then allocate $a$ to whichever agent trimmed it most. The other agents get their second most preferred complete piece.
						%	 \haris{clarify which non-cutter got a complete piece}
							 \STATE Give the cutter (agent 4) the remaining complete piece.
							 \ENDIF

									\RETURN envy-free allocation and any cake that is still not allocated.
								%	\algstore{myalg2}
					 \end{algorithmic}
					 \end{minipage}
					\end{algorithm}

	We make another observation about the outcome of the core protocol.
	
	\begin{lemma}\label{lemma:completepiece}
		During the core protocol, when an agent $i$ makes trims to equal his second or third most preferred piece respectively, $i$ either gets such a piece completely or some other agent gets such a piece completely that $i$ values as much as second or third most preferred piece respectively. 
		\end{lemma}
    	\begin{proof}
    		Assume agent $i$ is not the cutter in the core protocol and he was asked to trim his first and second most preferred pieces to equal his third most preferred piece. If agent $i$ gets the complete piece that is his third most preferred piece, we are already done. Let us say that he got a partial piece. Then, at most one other agent got a partial piece. This means that some other agent $j$ got a complete piece that is either agent $i$ first, second or third most preferred piece.

    		Assume agent $i$ is not the cutter in the core protocol and he was asked to trim his most preferred pieces to equal his second most preferred piece. This means that $i$ got either his most preferred piece up to the value of the second most preferred piece or he got the second most preferred piece completely. If he got the second most preferred piece completely, we are done. If $i$ got the most preferred piece up to the value of the second most preferred piece, then some other agent got a possibly partial piece from $i$'s most preferred piece. But note that since $i$ was asked to trim up to his second most piece in the protocol, only $i$ was competing for his second most preferred piece. Hence, some other agent $j$ got $i$'s second most preferred piece completely. 
    		\end{proof}
        
	% \begin{proof}
% 		Assume agent $i$ is not the cutter in the core protocol and he was asked to trim his first and second most preferred pieces to equal his third most preferred piece. If agent $i$ gets the complete piece that is his third most preferred piece, we are already done. Let us say that he got a partial piece. Then, at most one other agent got a partial piece. This means that some other agent $j$ got a complete piece that is either agent $i$ first, second or third most preferred piece.
%
%
% 		Assume agent $i$ is not the cutter in the core protocol and he was asked to trim his most preferred pieces to equal his second most preferred piece. This means that $i$ got either his most preferred piece up to the value of the second most preferred piece or he got the second most preferred piece completely. If he got the second most preferred piece completely, we are done. If $i$ got the most preferred piece up to the value of the second most preferred piece, then some other agent got a possibly partial piece from $i$'s most preferred piece. But note that since $i$ was asked to trim up to his second most piece in the protocol, only $i$ was competing for his second most preferred piece. Hence, some other agent $j$ got $i$'s second most preferred piece completely.
% 		\end{proof}

			The core protocol for partial envy-free allocation can be extended to obtain protocol for partial envy-free allocation with a single domination.
	If the core protocol is run again on the unallocated cake, we show in Lemma~\ref{lemma:single-domination-protocol} that the cutter dominates at least one agent. 
	%We will call such a protocol the \emph{single-domination protocol}. 
	When the core protocol is implemented, then the
			piece from which the highest valued (from the perspective of the cutter) residue is trimmed is called a \emph{significant piece}.
	%If the residues of both pieces that were not completely allocated have the same value for the cutter, then we say that there does not exist a significant piece. 
	 We will show that the cutter can be made to dominate an agent who got the significant piece. If residues from both pieces that are partially allocated are of the same value to the cutter, then we say that both non-cutters who got partial pieces were given significant pieces. In this case, a second run of the core protocol is enough for the cutter to dominate two agents.
			%The reason the piece is called significant is because the unallocated residue from this piece is of high enough value to the cutter so that the cutter may be concerned about being envious of another agent if a large enough portion of the residue is additionally given to an agent who did not get the significant piece. Of course, if the residue is given to the agent who got part of significant piece, the cutter is not envious.		
		%	We will observe that if both trimmings are of the same value for the cutter, that goes in our advantage because then cutter can be made to dominate both agents who get partial pieces. 

	%\end{minipage}

	% \haris{As soon as agents tell about the ordinal ranking of the pieces, even before the order of trims are elicited, it can be checked which agent will be able to get his second most or thirs most preferred piece. In this case we will refer to this being guaranteed a third or second most preferred piece EXPLAIN...}
 
	 \begin{lemma}[Single-Domination Protocol Lemma]\label{lemma:single-domination-protocol}
	For $n=4$, there exists a discrete and bounded protocol which returns an envy-free partial allocation in which one agent dominates another agent.
		 \end{lemma}
		 \begin{proof}
					% The core protocol for partial envy-free allocation can be extended to obtain protocol for partial envy-free allocation with a single domination.
	% 		If the core protocol is run two times, then the cutter dominates at least one agent.
	% 		We will call such a protocol the \emph{single-domination protocol}. To get single domination we use the procedure described in the two-complete pieces lemma.
	We run the core protocol a first time. This guarantees that the cutter (say agent 4) gets a complete piece (of value 1/4 of the whole cake) and that the residue is composed from the trims of at most 2 pieces. 
	Since from the cutter's perspective all pieces were equal, the residue cannot sum up to more than 1/2 of the cake for him. Recall that the piece which agent 4 thinks was trimmed is the significant piece. The total residue is composed of the residue from the significant piece (call the residue $\beta_1$) and the residue from the other trimmed piece (let us call this residue $\beta_2$). 
	The cutter thinks that he got $V_4(\beta_1)$ more value than the agent who got the significant piece. 
	Since $V_4(\beta_1)\geq V_4(\beta_2)$, the value of the total residue from the cutter's perspective is at most $2V_4(\beta_1)$. If we run the core protocol again, at most two pieces are partial and hence the residue's value for the cutter is at most $2\times 2V_4(\beta_1)/4=V_4(\beta_1)$. This implies that even if the agent who got the significant piece gets all the residue which is of value $V_4(\beta_1)$ to the cutter, the cutter would still not envy him. 
	This implies the cutter dominates the agent who got the significant piece.
	%We can run item 1 and 2 again so that we guarantee that from agent 1's perspective at least one agent can be assigned the whole of the residue and not have more than him. We say that agent 1 now dominates this agent.
			 \end{proof}
		
			Although the core protocol can be easily used to enable the cutter to dominate one agent, dominating two agents is more challenging. In the next section, we show how to overcome this challenge.

			\subsection{Permutation Protocol}
			If we run the core protocol repeatedly on the remaining unallocated cake, it may be that the cutter keeps dominating the same agent. We show that we only need to run the core protocol 5 times in total to achieve double domination. It may be that each time, the core protocol is run, the same agent gets the significant piece and hence the cutter dominates the same agent. 
	If a different agent gets a significant part of the residue in any of the iterations, then the double domination is already achieved with one more iteration since from the cutter's perspective we have 2 agents who may be given all that is left of the cake without him being envious of them. 
	If not, then we have one agent who ends up with the piece from which a significant trim was obtained for all iterations. The permutation lemma tells us that it is possible to give one of the 4 significant pieces to another agent while still preserving envy-freeness. This ensures that agent 1 ends up dominating two agents.

					\begin{algorithm}[h!]
					  \caption{Permutation Protocol for 4 Agents}
					  \label{algo:permutation}
					\renewcommand{\algorithmicrequire}{\wordbox[l]{\textbf{Input}:}{\textbf{Output}:}}
					 \renewcommand{\algorithmicensure}{\wordbox[l]{\textbf{Output}:}{\textbf{Output}:}}
			%		\algsetup{linenodelimiter=\,}
					\begin{algorithmic}
					%	\small
						\REQUIRE An outcome of a core protocol in which one agent (say agent 4) is the cutter and another specified agent (say agent 1) gets the significant piece.
						\ENSURE An allocation in which each agent gets a piece equal to the value he trimmed to in the core protocol and in which agent 2 or 3 gets the significant piece.
					\end{algorithmic}
					  \begin{algorithmic}[1]
						  %\normalfont
%	\small
						  \IF{agent 2 was competing with someone for the piece and therefore had a trimmed piece}
						\begin{enumerate}
							\item If the agent who made the second rightmost trim on agent 2's piece is agent 1, then we can simply permute agents 1 and 2 i.e., exchange their pieces.
						\item 	If the agent who made the second rightmost trim on agent 2's piece is agent 3, then we can move 3 to agent 2's piece. % Notice that 3 must have been allocated a complete piece since only 2 pieces may be trimmed. This means that 1 did not compete with 3 for 3's piece which implies that if 1 can get 4 or 3's piece he will get the piece he was guaranteed before the order of the trim was determined (either second most preferred or third most preferred).
						 Agent 2 can be given 1's (trimmed) piece.  Agent 1 can be given one of the complete pieces (which was given to 3 or 4). Agent 4 can be given the remaining complete piece.
						\end{enumerate}
						  \ELSIF{agent 2 was in possession of a complete piece for which he was not competing with another agent}
			%	  		Let us focus on the piece agent 1 desires and was guaranteed to get (either his second most preferred or third most preferred). We need to free this piece so that agent $1$ can take it. We distinguish between three subcases.

				  		\begin{enumerate}
				  		    \item If 2's piece is the piece such that agent 1 trimmed up to that value in the core protocol, then simply permute 1 and 2.
				  			\item If 4 is holding the piece such that agent 1 trimmed up to that value in the core protocol then we simply move 4 to 2's piece since it is a complete piece and agent 1 gets 4's piece. Agent 2 is given 1's piece.
							\item If  3 has a complete piece such that agent 1 trimmed up to that value in the core protocol and 3 is indifferent between two pieces among his top 3 pieces, then 3 can be given another complete piece (such that he trimmed up to the value of that piece) of either 2 and 4. Agent 1 can be given 3's piece. Agent 2 gets 1's piece and 4 gets the remaining complete piece.
				  		\end{enumerate}
						  \ENDIF
					 \end{algorithmic}
					\end{algorithm}

			% Since, our aim is that the cutter should dominate two agents, we try to reallocate (permute) some of the already allocated pieces to ensure that the cutter dominates a second agent. 
			%The idea of exchanging or \emph{permuting} some pieces of the cakes has been considered in the literature~\citep[see \eg][]{FeKi74a}.	
	We will use this idea in the argument of the Permutation Lemma. Before presenting the Permutation Lemma, we present another lemma that is useful for the proof of the Permutation Lemma.

				\begin{lemma}%[Wonki Cow Lemma]
					\label{lemma:wonkicow}
	Consider $m$ rows each with $m-1$ entries of positive reals. Then there exists at least one row such that for each entry in the row, the sum of other entries in the column corresponding to that entry is greater than or equal to the entry in the row.
				\end{lemma}
				\begin{proof}
	It is sufficient to find a row in which each entry is not the unique maximum entry for that column. We go row by row and eliminate a row if it has at least one entry that is a maximal value among all entries in the corresponding column. Even if $m-1$ rows are eliminated, we are left with one row in which each entry is not the unique maximum entry for that column.  
					\end{proof}

				We will rely on Lemma~\ref{lemma:wonkicow} while reasoning about when reallocating pieces does not cause any envy. In particular let us say that an agent gets slightly more than the value he wanted to guarantee. He may not want to let go of this extra value lest it leads to him being envious. However, let us say he gets similar extra values again, then we may ask the agent to choose which one of the extra values he rates least and give this extra value to some other agent. The other extra values, make up for this loss. Intuitively, Lemma~\ref{lemma:wonkicow} will help identify that if we have enough subcases (rows) then there will be one row on which an agent will be happy to compromise.

	We are now in a position to present the Permutation Lemma.

		\begin{lemma}[Permutation Lemma]\label{lemma:perm}
		There exists a discrete and bounded protocol for 4 agents that returns a partial envy-free allocation in which one agent dominates two other agents.
			\end{lemma}
					\begin{proof}
						Assume that agents 1, 2, 3, 4 get pieces $p_1$, $p_2$, $p_3$, and $p_4$ respectively with $4$ getting complete pieces. When $4$ cut the pieces, $p_1$ is a piece $P_1$ without the part left of the second rightmost trim. 
						Now, assume that in four iterations of the $1$ gets the significant piece and we want to reallocate so that some other agent among $2$ and $3$ gets the significant piece. We need to show that when the reallocation is done then barring the bonus part that agents get in the core protocol, the agents get at least as preferred a piece.
						If another agent aside from agent 1 gets the significant piece then we are done. If agent 1 is repeatedly getting it, then we zoom in to a case to permute it i.e., reallocate some of the pieces so as to make sure that an agent other than 1 is dominated. If agent $1$ is repeatedly getting the significant piece and there is no reallocation in which instead of agent $1$ some other agent gets the significant piece, this means that either (a) agent $1$ likes the bonus from his significant piece so much that he is not willing to take some other piece or (b) for some other agent $j$, the bonus corresponding to the significant piece is not enough for $j$ to be attracted towards the significant piece. 
					
						\begin{claim}\label{claim:permutation}
							For a partial allocation as a result of the core protocol in which agents make trims, a reallocation can be done in which some agent other than $1$ gets the significant piece and each agent gets a piece of value corresponding to his original trims but in which he may lose out on the additional bonus due to the second rightmost trims. 
							\end{claim}
                    			\begin{proof}
                    				The algorithm to perform the reallocation is stated as the Permutation Protocol~(Algorithm~\ref{algo:permutation}).
                    		Imagine that agent 1 is holding the significant piece. For the piece to be significant, another agent must have been competing with agent 1 for it. If no other agent was competing for the piece, then agent 1 would have got the whole piece and hence the piece would not be significant.
                    		%		\haris{what if this piece is the second most preferred piece of the agent 1, Then is it necessary that 1 gets this piece? Need to argue and and give reference to this property of the core protocol }

                    				 Let us say that the agent who competes for the same piece is agent 2. Since the piece was cut up to the second rightmost trim,
                    				agent 2 is not envious of any other agent if agent 2 gets the piece (while not getting his trim).
                    				%\haris{without the trim or with the }

                    				We now distinguish between 2 possibilities. 
                    				% \begin{enumerate}
                    		% 			\item Agent 2 was competing with someone for the piece and therefore had a trimmed piece.
                    		% 			\item Agent 2 was in possession of a complete piece for which he was not competing with another agent.
                    		% 			\end{enumerate}
                    				We show how both cases can be handled. 

                    				\begin{enumerate}
                    					\item \textbf{Agent 2 was competing with someone for the piece and therefore had a trimmed piece.} We distinguish between two subcases.
                    					\begin{enumerate}
                    						\item If the agent who made the second rightmost trim on agent 2's piece is agent 1, then we can simply permute agents 1 and 2 i.e., exchange their pieces.
                    					\item 	If the agent who made the second rightmost trim on agent 2's piece is agent 3, then we can move 3 to agent 2's piece. Notice that 3 must have been allocated a complete piece since only 2 pieces may be trimmed. This means that 1 did not compete with 3 for 3's piece which implies that if 1 can get 4 or 3's piece he will get the piece he was guaranteed before the order of the trim was determined (either second most preferred or third most preferred). Agent 2 can be given 1's (trimmed) piece.  Agent 1 can be given one of the complete pieces (which was given to 3 or 4). Agent 4 can be given the remaining complete piece. 
                    					\end{enumerate}

                    					\item \textbf{Agent 2 was in possession of a complete piece for which he was not competing with another agent.} 
                    				Let us focus on the piece agent 1 desires and was guaranteed to get a piece of that value (either his second most preferred or third most preferred). %By Lemma~\ref{lemma:completepiece}, some other agent got such a piece completely. We need to free this piece so that agent $1$ can take it. 
                    		 We will use the fact that by Lemma~\ref{lemma:completepiece}, some other agent got a complete piece that was of at least as much value to agent 1.
                    		 We distinguish between three subcases.

                    				\begin{enumerate}
                    				    \item If 2's piece is the piece agent 1 desires, then simply permute 1 and 2.
                    					\item If 4 is holding the piece to be freed then we simply move 4 to 2's piece since it is a complete piece and agent 1 gets 4's piece. Agent 2 is given 1's piece. 
                    					\item Let us assume that 3 is holding the piece to be freed and agent 1 is guaranteed to get a piece of a value as much as
                    				this complete piece (which we refer to as $a$). 
		
                    		%		If 3 put a trivial trim on his piece. 
		
                    		% \begin{enumerate}
                    		% 	\item

                    		\textbf{If 3 is indifferent among any two of his most preferred three pieces}, then either he is indifferent among the top 2 or among the second and third. For the former, if he got a top 2 piece, he will be willing to get another top 2 piece and if he got a third piece, he would certainly be happy to get one of the complete top 2 pieces. For the latter, if he got one of the second or third preferred pieces, he would be happy to get the other equally preferred one. If he had got the top piece, then this means that only 3 had a proper trim on $a$. But this means that $3$ is willing to move to another complete piece with as much value barring the extra bonus he got in piece $a$. Hence in all the cases above, 3 can be given another complete piece, agent 1 can be given 3's complete piece, and 2 can get 1's partial piece. \\

                    		 \textbf{We now assume that 3 is not indifferent among any of his 3 most preferred pieces}.  We show that in this scenario, we would already be in one of the previous cases in which 1 is okay with taking 2 or 4's piece.
		
                    		%		\haris{why is it complete? does that mean three pieces were complete? 2,3,4?}
		
                    				 % or a piece of at least as much value. 
                    				%An agent is guaranteed to get a piece if he has the rightmost trim for that piece. This means that either he can get that piece or another piece of as much value.	
		
                    				We first argue that agent $3$'s most preferred piece is the significant piece.
                    				Piece $a$ cannot be agent 3's most preferred piece since otherwise $1$ would not be guaranteed to get it: if $1$ got it, 3 would be envious. 
                    			%	\haris{explain guaranteed to get!} 
                    				Neither can the piece held by 4 or 2 be 3's most preferred piece since envy-freeness of the partial allocation will be violated. 
                    				%HARIS: I commented out the next sentence because I think it was not needed.
                    				%Agent $2$ has a complete piece, which means that agent $3$ was not competing for it or he would have trimmed it. 
                    				Since the pieces allocated to 2, 3 and 4 are not $3$'s most valued pieces, it implies that
                    				agent $3$'s most preferred piece is the significant piece.

                    		Next, we argue that 3 was allocated his second most preferred piece in which had put a proper trim. First we assume that 3 trimmed up to his third most preferred piece. If $a$ is 3's third most preferred piece, then $3$ would have been envious of either 2 or 4 whoever got 3's most preferred piece. 3 would have competed for such as piece and such a piece would not have been allocated completely. This means that 3 got his second most preferred piece. Second, we assume that 3 trimmed unto his second most preferred piece, then this means 3 was allocated his second most preferred piece since he was guaranteed it. In both cases 3 put a proper trim on $a$.

                    		We now argue that $a$ (3's piece) is also agent 1's second most preferred piece. If the piece allocated to 2 or 4 was agent 1's second most preferred piece, we would be in a different case. 

                    		Therefore $a$ is agent 1's and 3's second most preferred piece. 
                    		Since $1$ is guaranteed to get a piece of value of $a$ and since $a$ is agent $1$'s second most preferred piece, it means that only $1$ has a proper trim on $a$. But we have already shown that $3$ also put a proper trim on $a$. But if two agents put a proper trim on a piece, the piece cannot be completely allocated to an agent hence a contradiction. 
                    	\end{enumerate}
                    				\end{enumerate}
                    				This complete the proof of the claim. 
                    						\end{proof}

			We have shown that reallocation can be done in which each agent gets a piece of value corresponding the trims the agents made. 
	
		Note that in each iteration of protocol, each agent gets some (possibly non-zero) bonus. Since agent $1$ gets a significant piece, he gets a non-zero bonus each time. Although $1$ may object to any reallocation for a particular iteration of core protocol, we show that he will not envious if one particular iteration of the core protocol is identified in which corresponds to the row of entries in Table~\ref{table:bonus} in which for each entry in the row, the sum of other entries in the column corresponding to that entry is greater than or equal to the entry in the row. By Lemma~\ref{lemma:wonkicow}, such a row exists. This means that even if each agent loses his bonus because of the permutation, he gets enough bonus in the other iterations to make for this loss so that there is no envy.
	
		We have shown so far that in four iterations of the core protocol, the partial envy-free allocation is such that two different agents got a significant piece in one of the calls of the core protocol. If the same agents get a significant piece in all the first four iterations of the core protocol, then we have shown above that the permutation protocol can implemented to give the significant piece to another agent and still not maintain envy-free across the four iterations of the core protocol.
	The fifth iteration of the core protocol ensures that the cutter dominates two agents and the partial allocation is envy-free.
			\qed \end{proof}
\vspace{-1em}
					\begin{table}[h!]
						\centering
				%	\small
					\scalebox{0.9}{
					\centering
				\label{tab:compare}
				\begin{tabular}{lcccc}
				\toprule
				&1&2&3&4\\
				\midrule
				Bonus in 1st iteration&$b_1^1$&$b_2^1$&$b_3^1$&0\\
				Bonus in 2nd iteration&$b_1^2$&$b_2^2$&$b_3^2$&0\\
				Bonus in 3rd iteration&$b_1^3$&$b_2^3$&$b_3^3$&0\\
				Bonus in 4th iteration&$b_1^4$&$b_2^4$&$b_3^4$&0\\
				\bottomrule
				\end{tabular}
			}
				\caption{Bonus of each agent in the first 4 calls of the core protocol with agent 4 as the cutter.}
			%		\caption{Properties satisfied by SDSs.}
				% \caption{Efficiency and strategyproofness notions satisfied by RSD and \PS. }
				\label{table:bonus}
				\end{table}

	\subsection{Overall Protocol}
	In the previous sections, we built the building blocks for our overall protocol: Post Double Domination Protocol, Core Protocol, and Permutation Protocol. We are now in a position to formalize the protocol to compute a complete envy-free allocation and presents the main result. The protocol is formalized as Algorithm~\ref{algo:ef4agents}.

				\begin{algorithm}[h!]
				  \caption{Discrete and Bounded Envy-free Protocol for 4 Agents.}
				  \label{algo:ef4agents}
				\renewcommand{\algorithmicrequire}{\wordbox[l]{\textbf{Input}:}{\textbf{Output}:}} 
				 \renewcommand{\algorithmicensure}{\wordbox[l]{\textbf{Output}:}{\textbf{Output}:}}
		%		\algsetup{linenodelimiter=\,}
				%\begin{algorithmic}
				%	\REQUIRE 
				%	\ENSURE EF complete allocation.
				%\end{algorithmic}
				  \begin{algorithmic}[1] 
					  %\normalfont
				%	  \small
					  \WHILE{some agent $i$ does not dominate two other agents and there is still some unallocated cake}
					  %\STATE Single domination protocol: \begin{enumerate} 
					  \STATE Run the Core Protocol (Algorithm~\ref{algo:core}) 4 times on the unallocated cake with $i$ as the cutter. Return at any point if there is no cake left unallocated.
					  \COMMENT{After two iterations, $i$ already dominates one agent}
					 %  	\item  Run core protocol with $i$ as cutter.
					 % \item Run core protocol on the remains with $i$ as cutter. %\COMMENT{By now $i$ already dominates at least one agent. These first two steps are referred to as the single domination protocol.}
					 %  \end{enumerate}
					 %
					 %
					 %  \STATE Run the single domination protocol with $i$ as the cutter $3c_1+1$ times. If no more cake is left at any point, we are already done.
					 % \STATE When single domination protocol is run so many times, there is a subcase (which combination of agents have what ordering on the pieces) that is repeated 4 times.
					 \IF{the same agent (say agent $j$) gets the significant piece in each of the 4 calls of the core protocol}
					 \STATE Identify in which call of the protocol, the non-cutter agents get less bonus than the sum of bonuses in the other calls of the core protocol. \COMMENT{see Table~\ref{table:bonus}}
					 \STATE Implement reallocation via the Permutation Protocol (Algorithm~\ref{algo:permutation}) for pieces allocated by this particular call of the core protocol where $i$ is the cutter and $j$ gets the significant piece. 
					 \ENDIF
					 \STATE Run core protocol in the unallocated cake if some cake is still unallocated.
					% \COMMENT{When the core protocol is run 5 times with $i$ as the cutter, then it is possible to implement permutation reallocation so as to make sure $i$ dominates at least two agents.}
					 % \STATE Implement a reallocation of the allocated cake so that $i$ dominates two agents.
					  \ENDWHILE
					  \IF{there is some unallocated cake}
					  \STATE Run Post Double Domination Protocol (Algorithm~\ref{algo:postDD}) on the remaining cake.
					  \ENDIF
					  		\RETURN envy-free complete allocation.		
				 \end{algorithmic}
				\end{algorithm}

	\begin{theorem}\label{th:finaltheorem}
		For four agents, there exists a discrete and bounded envy-free protocol that requires constant number of queries in the Robertson and Webb model. 
	\end{theorem}
	\begin{proof}

		The protocol is formalized as Algorithm~\ref{algo:ef4agents}.
        The theorem follows from Lemmas~\ref{lemma:dd} and \ref{lemma:perm}. %We first observe that each step in the overall protocol including the core protocol is discrete. 
	%We next argue that our protocol returns a complete and envy-free allocation. 
	The protocol first ensures that there exists a partial envy-free allocation in which each agent dominates two other agents. 
	The protocol then used
	the Double Domination Protocol to obtain a complete envy-free allocation. 
	%The protocol required 203 cuts and 584 Robertson-Webb query operations.
    % We now show that the protocol is bounded not only in terms of number of cuts but also in terms of Robertson-Webb queries.
    
We now argue in more detail that the protocol is bounded not only in terms of number of cuts but also in terms of Robertson-Webb queries.
		
In each run of the core protocol, the cutter makes three cuts and then each non-cutter makes at most two trims. Two agents may be asked to replace their cuts and make a cut to make their most preferred piece equal to their second most preferred piece or in another case agents perform Divide and Choose which requires 1 cut.
	Hence, in the core protocol,  there are at most 3+6+1=11 cuts. Since we run 5 iterations of the core protocol to get one double domination, we need 50 cuts to get one double domination. Since we need each agent to double dominate, we need $50\times 4=200$ cuts. After we achieve double dominations for all agents, we need at most three more cuts. So we need at most 203 cuts.

	We now count the maximum number of Robertson-Webb query operations required during the protocol. The cutter is asked the value of the whole cake ($1$ query), and then made to cut to equal 1/4 of the value ($3$ queries). The non-cutter agents are asked to give values of the pieces ($3\times 4=12$ queries.) They are then asked to trim to equal the value of the third most preferred piece ($3\times 2$ queries). Two agents may be asked to make one additional cut to make their most preferred piece equal to their second most preferred piece ($2\times 1$ queries) or in another case, agents perform Divide and Choose which requires $4$ queries.
	Hence the core protocol requires in total $1+3+12+6+4=26$ queries.
	The core protocol is run 5 times so that takes $26\times 5$ queries.
	In addition to running the core protocol 5 times, we also need to run permutation protocol four times to check which permutation to implement. In each call of the permutation protocol, agents are queried about the amount of bonus which they get in the core protocol which takes additional 3 queries so in total of 12 queries for the permutation protocol.
	In order to get all double dominations, we require $((26\times 5) + 12)\times 4=568$ queries. After we achieve double dominations,  we require maximum 16 more queries. So the total number of queries is 584.
		\end{proof}
    % \begin{proof}
    %     The protocol is formalized as Algorithm~\ref{algo:ef4agents}.
    % We first observe that each step in the overall protocol including the core protocol is discrete.
    % We next argue that our protocol returns a complete and envy-free allocation.
    % The protocol first ensures that there exists a partial envy-free allocation in which each agent dominates two other agents.
    % The protocol then implements the steps in the proof of
    % the Double Domination Lemma to obtain a complete envy-free allocation.
    % The protocol requires at most 203 cuts and 584 Robertson-Webb query operations.
    %     \end{proof}

    \section{Discussion}
    \label{section:disc}

    In this paper, we proposed the first bounded and envy-free protocol for four agents. % This means that in addition to the Selfridge-Conway protocol which works perfectly well for the three musketeers sharing a cake, we now have a bounded protocol in case the fourth musketeer D'Artagnan joins Porthos, Athos and Aramis.
    Some of our insights  such as the one of exploiting the bonus cake given to the agents as well as analyzing the domination graph may be useful in attacking the problem for \emph{any} number of agents.
Our protocol is based on three main ingredients: core protocol, permutational protocol, and post double domination protocol. The higher level ideas of our overall protocol could be applied to general problem for more number of agents. 
Some of the proof ideas may be generalizable to more than four agents. For example, a suitable generalization of the core protocol is feasible. On the other hand, the Permutation Protocol appears to be challenging to extend to arbitrary number of agents and will require more interesting insights and techniques. There are also some other interesting directions for future research such as existence and properties of equilibria under our new protocol.
    %We leave it as interesting open problem to formulate a general bounded and envy-free protocol.

%However the permutation protocol will need to be generalized in a significant manner for the general problem.    
    
    % Although some of the proof ideas may be generalizable to more than four agents, the arguments in the proof of the Permutation Lemma appear to be challenging to extend to arbitrary number of agents.
    %We leave it as interesting open problem to formulate a general bounded and envy-free protocol.

		\section{Acknowledgments}
	Data61 is funded by the Australian Government through the Department of Communications and the Australian Research Council through the ICT Centre of Excellence Program. Part of the research was carried out  when the authors were visiting LAMSADE, University Paris Dauphine  in May 2015. The authors thank Rediet Abebe, Steven Brams, Simina Br{\^{a}}nzei, Ioannis Caragiannis, Katar{\'{\i}}na Cechl{\'{a}}rov{\'{a}}, Serge Gaspers, David Kurokawa, and Ariel Procaccia for useful comments. 
They also thank the reviewers of STOC 2016 for suggestions to improve the presentation.   Haris Aziz thanks Ulle Endriss for introducing the subject to him at the COST-ADT Doctoral School on Computational Social Choice, Estoril, 2010.

\bibliographystyle{abbrv}
\balance

%\renewcommand*{\bibfont}{\small}
       % \bibliography{../../../research/papers/pamas/abb,../../../research/papers/pamas/pamas,../../../research/papers/pamas/aziz,../../../research/papers/pamas/brandt}
% % %
%

\end{document}